\documentclass[a4paper,ngerman,english]{article}
\usepackage{mathpazo}
\usepackage[T1]{fontenc}
\usepackage[latin9]{inputenc}
\usepackage{color}
\usepackage{amsthm}
\usepackage{amsmath}
\usepackage{amssymb}

\makeatletter

\newenvironment{lyxlist}[1]
{\begin{list}{}
{\settowidth{\labelwidth}{#1}
 \setlength{\leftmargin}{\labelwidth}
 \addtolength{\leftmargin}{\labelsep}
 }}
{\end{list}}
\theoremstyle{plain}
\newtheorem{thm}{\protect\theoremname}

\allowdisplaybreaks[0]

\newcounter{meg}

\newcommand{\R}{%
\refstepcounter{meg}%
{\vskip6pt \noindent \bf Remark \themeg.\ \ \ }%
}

\makeatother

\usepackage{babel}
  \addto\captionsenglish{\renewcommand{\theoremname}{Theorem}}
  \addto\captionsngerman{\renewcommand{\theoremname}{Theorem}}
\providecommand{\theoremname}{Theorem}

\begin{document}

\title{Operational understanding of the covariance of classical electrodynamics}

\author{Márton Gömöri%
\thanks{gomorim@gmail.com%
} {\normalsize ~~and}~~László E. Szabó\emph{\small }%
\thanks{leszabo@phil.elte.hu%
}\emph{\small }\\
\emph{\small ~}\\
\emph{\small Department of Logic, Institute of Philosophy}\\
\emph{\small Eötvös University, Budapest}}

\date{{\normalsize Draft of May 2013}}

\maketitle
\setlength{\unitlength}{1cm} 

\begin{picture}(0,2.5)(2,-10.5) \put(0,0){\begin{minipage}[t]{1\columnwidth} {\footnotesize  Published as }\emph{\footnotesize Physics Essays} {\bf\footnotesize 26 }{\footnotesize (2013) pp. 361--370. DOI: 10.4006/0836-1398-26.3.361} \end{minipage}}\end{picture}
\begin{abstract}
It is common in the literature on classical electrodynamics and relativity
theory that the transformation rules for the basic electrodynamic
quantities are derived from the pre\emph{-}assumption that the equations
of electrodynamics are covariant against these---unknown---transformation
rules. There are several problems to be raised concerning these derivations.
This is, however, not our main concern in this paper. Even if these
derivations are regarded as unquestionable, they leave open the following
fundamental question: Are the so-obtained transformation rules indeed
identical with the true transformation laws of the empirically ascertained
electrodynamic quantities? 

This is of course an empirical question. In this paper, we will answer
this question in a purely theoretical framework by applying what J.
S. Bell calls \textquotedblleft{}Lorentzian pedagogy\textquotedblright{}---according
to which the laws of physics in any one reference frame account for
all physical phenomena, including what a moving observer must see
when performs measurement operations with moving measuring devices.
We will show that the real transformation laws are indeed identical
with the ones obtained by presuming the covariance of the equations
of electrodynamics, and that the covariance is indeed satisfied. Beforehand,
however, we need to clarify the operational definitions of the fundamental
electrodynamic quantities. As we will see, these semantic issues are
not as trivial as one might think.
\end{abstract}

\thispagestyle{empty}

\vfill{}

\newpage{}

\section{Introduction\label{sec:Introduction}}

Consider two inertial frames of reference $K$ and $K'$. Let $\left(x,y,z,t,\mathbf{E},\mathbf{B},\varrho,\mathbf{j}\right)$
denote the basic physical quantities involved in electrodynamics,
that is the space and time coordinates, the electric and magnetic
field strengths, and the source densities, obtainable by means of
measuring equipments \emph{co-moving with} $K$. Let $\left(x',y',z',t',\mathbf{E}',\mathbf{B}',\varrho',\mathbf{j}'\right)$
be the same quantities in $K'$, that is, the quantities obtainable
by means of the same operations with the same measuring equipments
when they are \emph{co-moving with} $K'$\emph{. }

By \emph{transformation law} we mean a one-to-one functional relation,
\[
T:\left(x,y,z,t,\mathbf{E},\mathbf{B},\varrho,\mathbf{j}\right)\mapsto\left(x',y',z',t',\mathbf{E}',\mathbf{B}',\varrho',\mathbf{j}'\right)=T\left(x,y,z,t,\mathbf{E},\mathbf{B},\varrho,\mathbf{j}\right)
\]
expressing the law-like regularity that if in an arbitrary space-time
point $A$ the $K$-quantities take values $\left(x(A),y(A),z(A),t(A),\mathbf{E}(A),\mathbf{B}(A),\varrho(A),\mathbf{j}(A)\right)$
then, in the same space-time point $A$, the corresponding $K'$-quantities
take values 
\begin{eqnarray}
\left(x'(A),y'(A),z'(A),t'(A),\mathbf{E}'(A),\mathbf{B}'(A),\varrho'(A),\mathbf{j}'(A)\right)\nonumber \\
=T\left(x(A),y(A),z(A),t(A),\mathbf{E}(A),\mathbf{B}(A),\varrho(A),\mathbf{j}(A)\right)\label{eq:transformation-def}
\end{eqnarray}
and vice versa. 

A system of equations is said to be \emph{$T$-covariant}, that is,
covariant against this transformation law, if expressing the variables
$\left(x,y,z,t,\mathbf{E},\mathbf{B},\varrho,\mathbf{j}\right)$ in
the equations by means of $\left(x',y',z',t',\mathbf{E}',\mathbf{B}',\varrho',\mathbf{j}'\right)$
we obtain a system of equations of exactly the same form in the primed
variables as the original one in the original variables.

One cannot a priori assume that there exists a transformation law
in the above sense; the fact that there is a law-like connection between
the quantities in $K$ and in $K'$ at all is a \emph{contingent}
fact of the physical world. In particular, as it turns out, the space-time
coordinates $\left(x'(A),y'(A),z'(A),t'(A)\right)$ are completely
determined by the space-time coordinates $\left(x(A),y(A),z(A),t(A)\right)$,
the field strengths $\left(\mathbf{E}'(A),\mathbf{B}'(A)\right)$
by the field strengths $\left(\mathbf{E}(A),\mathbf{B}(A)\right)$,
and the source densities $\left(\varrho'(A),\mathbf{j}'(A)\right)$
by the source densities $\left(\varrho(A),\mathbf{j}(A)\right)$,
separately. That is to say, the transformation law (\ref{eq:transformation-def})
consists of three maps:
\begin{eqnarray}
\left(x'(A),y'(A),z'(A),t'(A)\right) & = & T_{1}\left(x(A),y(A),z(A),t(A)\right)\label{eq:elso-trafo}\\
\left(\mathbf{E}'(A),\mathbf{B}'(A)\right) & = & T_{2}\left(\mathbf{E}(A),\mathbf{B}(A)\right)\label{eq:masodik-trafo}\\
\left(\varrho'(A),\mathbf{j}'(A)\right) & = & T_{3}\left(\varrho(A),\mathbf{j}(A)\right)\label{eq:hatmadik-trafo}
\end{eqnarray}

As to the space-time coordinates, we take it for granted that the
functional relation (\ref{eq:elso-trafo}) is the well-known Lorentz
transformation (see Appendix~1). However, the Lorentz transformation,
and the transformation laws of other kinematic quantities derived
from it, alone, does not determine the transformation law of the electrodynamic
quantities. In the literature on classical electrodynamics and relativity
theory the transformation laws $T_{2}$ and $T_{3}$ in (\ref{eq:elso-trafo})--(\ref{eq:hatmadik-trafo})
are derived from the \emph{additional assumption} that the equations
of electrodynamics are covariant against these transformation laws---in
conjunction with the Lorentz transformation $T_{1}$. Among those
with which we are acquainted, there are basically two major versions
of these derivations, which are briefly summarized in the Appendix~2.
There are several problems to be raised concerning these derivations,
and certain steps are questionable. This is however not our main concern
in this paper. For, even if these derivations are regarded as unquestionable,
they only prove what the transformation laws $T_{2}$ and $T_{3}$
\emph{should} look like in order that the equations of electrodynamics
constitute a covariant system of equations with respect to these transformations.
But they leave open the question whether the so-obtained transformations
are indeed identical with the \emph{true} transformation laws; whether
it is indeed the case that the values obtained from $\mathbf{E}(A),\mathbf{B}(A),\varrho(A),\mathbf{j}(A)$
by means of the transformation rules we derived are equal to the \emph{real}
$\mathbf{E}'(A),\mathbf{B}'(A),\varrho'(A),\mathbf{j}'(A)$, that
is, the quantities obtained by the same operations with the same measuring
equipments when they are co-moving with\emph{ $K'$,} in the same
space-time point $A$. The obvious problem is that there does not
exist, and, in fact, it is hard to imagine, an independent confirmation
of the covariance of the equations---against an unknown transformation
law. That is, to confirm that the equations of electrodynamics \emph{really}
satisfy the requirement of covariance, we need a primary knowledge
of the transformation laws.

It must be emphasized that the requirement of covariance, as a necessary%
\footnote{The proper relationship between the relativity principle and covariance
is a subtle issue even in the context of special relativity (Bell~1987;
Norton~1993; Gr\o n and V\o yenli 1999; Gömöri and Szabó~2011).%
} condition for satisfying the special relativity principle, does not
simply mean formulating the laws of electrodynamics in some invariant
mathematical form, for example, as Lorentz tensor equations; the equations
must be covariant against the \emph{real physical} transformation
laws. The same point is emphasized by Gr\o n and V\o yenli (1999,
p.~1731) in the context of the generalized principle of relativity: 
\begin{quote}
All quantities appearing in a covariant equation, must be defined
in the same way in every coordinate system, and interpreted physically
without reference to any preferred system. {[}$\ldots${]} A law fulfilling
the restricted covariance principle, has the same mathematical form
in every coordinate system, and it expresses a physical law that may
be formulated by the same words (without any change of meaning) in
every reference frame {[}$\ldots${]}
\end{quote}
In our understanding, the ``meanings of the words'' by which a physical
law is formulated are determined by the empirical/operational definitions
of the quantities appearing in the law. Our considerations, therefore,
will be based on the operational definitions of the electrodynamic
quantities---this is an essential feature of our approach.

Throughout the paper we use the traditional 3+1 vector-analytic formulation
of the laws of electrodynamics. The reason is that this formalism
is convenient for our main purpose: to ascertain the true transformation
laws of the electrodynamic quantities in any one empirically verifiable
form. This problem is epistemologically \emph{prior} to the problem
of the proper algebraic/geometric interpretation of these transformation
laws. For, once we \emph{know} these laws in any one available form,
we can think about the best mathematical representation of them. (For
a current discussion of the various mathematical formulations, see
Ivezi\'{c}~2001, 2003; Hestenes~1966, 2003; Huang~2008, 2009;
Arthur~2011.)

Thus, what are the true transformation laws of the fundamental electrodynamic
quantities? This is of course an empirical question, which we are
not able to answer in this paper. Instead, we investigate the following
\emph{theoretical} question: Are the transformation rules derived
from the pre\emph{-}assumption of the covariance of the laws of electrodynamics
\emph{consistent} with the laws of electrodynamics in any single frame
of reference? In other words:
\begin{lyxlist}{00.00.0000}
\item [{(Q)}] What do the transformation laws (\ref{eq:masodik-trafo})--(\ref{eq:hatmadik-trafo})
look like in the prediction of the laws of physics in any single frame
of reference?
\end{lyxlist}
The basic idea is what J. S. Bell (1987) calls ``Lorentzian pedagogy'',
according to which ``the laws of physics in any one reference frame
account for all physical phenomena, including the observations of
moving observers''. That is to say, the laws of physics that are
valid in any one reference frame, say $K$, must account for the behaviors
of the moving measuring equipments and the results of all measuring
operations; therefore, must provide an answer to question (Q). 

The answer can be given by the laws of physics only if the question
is properly formulated. We must clarify what measuring equipments
and etalons are used in the empirical definitions of the electrodynamic
quantities; and we must be able to tell when two measuring equipments
are the same, except that they are moving, as a whole, relative to
each other---one is at rest relative to $K$, the other is at rest
relative to $K'$. Similarly, we must be able to tell when two operational
procedures performed by the two observers are the ``same'', in spite
of the prima facie fact that the procedure performed in $K'$ obviously
differs from the one performed in $K$. In order to compare these
procedures, first of all, we must know what the procedures exactly
are. All in all, a correct answer to question (Q) can be given only
on the bases of \emph{a coherent system of precise operational definitions}
of the quantities in question; and all these definitions must be represented
in the language of electrodynamics in a single frame of reference.
Interestingly, there is no explicit discussion of these issues in
the standard literature on electrodynamics and special relativity;
although, as we will see, none of these issues are as trivial as one
might think. 

Thus, accordingly, in the first part of the paper we clarify the operational
definitions of the electrodynamic quantities and formulate what electrodynamics
in a single inertial frame of reference---let us call it ``rest''
frame---exactly asserts in terms of the quantities so defined. In
the second part, applying the ``Lorentzian pedagogy'', on the basis
of the laws of electrodynamics in the ``rest'' frame, we derive
what a moving observer must see in terms of the ``rest'' frame quantities
when repeats the same operational procedures in the ``moving'' frame.
In this way, we obtain the transformation laws of the electrodynamic
quantities; that is to say, we derive the transformation laws from
the precise operational definitions of the quantities and from the
laws of electrodynamics in a single inertial frame of reference, \emph{without
of the pre}-\emph{assumption that the equations are covariant} against
these transformation laws---by which we answer our question (Q).

\section{Operational definitions of electrodynamic quantities in $K$}

In this section we give the operational definitions of the fundamental
quantities of electrodynamics (ED) in a single reference frame $K$
and formulate a few basic observational facts about these quantities. 

The operational definition of a physical quantity requires the specification
of \emph{etalon} physical objects and standard physical processes
by means of which the value of the quantity is ascertained. In case
of electrodynamic quantities the only ``device'' we need is a point-like
test particle, and the standard measuring procedures by which the
kinematic properties of the test particle are ascertained. 

So, assume we have chosen an \emph{etalon} test particle, and let
$\mathbf{r}^{etalon}(t)$, $\mathbf{v}^{etalon}(t)$, $\mathbf{a}^{etalon}(t)$
denote its position, velocity and acceleration at time $t$. It is
assumed that we are able to set the \emph{etalon} test particle into
motion with arbitrary velocity $\mathbf{v}^{etalon}<c$ at arbitrary
location. We will need more ``copies'' of the \emph{etalon} test
particle:

\paragraph*{Definition~(D0)}

A particle $e$ is called \emph{test particle} if for all $\mathbf{r}$
and $t$
\begin{equation}
\mathbf{v}^{e}\left(t\right)\biggl|_{\mathbf{r}^{e}\left(t\right)=\mathbf{r}}=\mathbf{v}^{etalon}\left(t\right)\biggl|_{\mathbf{r}^{etalon}\left(t\right)=\mathbf{r}}
\end{equation}
implies
\begin{equation}
\mathbf{a}^{e}\left(t\right)\biggl|_{\mathbf{r}^{e}\left(t\right)=\mathbf{r}}=\mathbf{a}^{etalon}\left(t\right)\biggl|_{\mathbf{r}^{etalon}\left(t\right)=\mathbf{r}}
\end{equation}
(The ``restriction signs'' refer to \emph{physical} situations;
for example, $|_{\mathbf{r}^{e}\left(t\right)=\mathbf{r}}$ indicates
that the test particle $e$ is at point $\mathbf{r}$ at time $t$.)
\medskip{}

\noindent Note, that some of the definitions and statements below
require the existence of many test particles; which is, of course,
a matter of empirical fact, and will be provided by (E0) below.

First we define the electric and magnetic field strengths. The only
measuring device we need is a test particle being at rest relative
to $K$.

\paragraph*{Definition~(D1)}

\noindent \emph{Electric field strength} at point $\mathbf{r}$ and
time $t$ is defined as the acceleration of an arbitrary test particle
$e$, such that $\mathbf{r}^{e}(t)=\mathbf{r}$ and $\mathbf{v}^{e}(t)=0$:
\begin{equation}
\mathbf{E}\left(\mathbf{r},t\right)\overset{^{def}}{=}\left.\mathbf{a}^{e}(t)\right|_{\mathbf{r}^{e}(t)=\mathbf{r};\,\mathbf{v}^{e}(t)=0}\label{eq:E-def}
\end{equation}
\medskip{}

\noindent Magnetic field strength is defined by means of how the acceleration
$\mathbf{a}^{e}$ of the rest test particle changes with an infinitesimal
perturbation of its state of rest, that is, if an infinitesimally
small velocity $\mathbf{v}^{e}$ is imparted to the particle. Of course,
we cannot perform various small perturbations simultaneously on one
and the same rest test particle, therefore we perform the measurements
on many rest test particles with various small perturbations. Let
$\delta\subset\mathbb{R}^{3}$ be an arbitrary infinitesimal neighborhood
of $0\in\mathbb{R}^{3}$. First we define the following function:
\begin{eqnarray}
\mathbf{U}^{\mathbf{r},t} & : & \mathbb{R}^{3}\supset\delta\rightarrow\mathbb{R}^{3}\nonumber \\
 &  & \mathbf{U}^{\mathbf{r},t}(\mathbf{v})\overset{^{def}}{=}\left.\mathbf{a}^{e}(t)\right|_{\mathbf{r}^{e}(t)=\mathbf{r};\,\mathbf{v}^{e}(t)=\mathbf{v}}\label{eq:U-def}
\end{eqnarray}
Obviously, $\mathbf{U}^{\mathbf{r},t}(0)=\mathbf{E}\left(\mathbf{r},t\right)$.

\paragraph*{Definition~(D2)}

\noindent \emph{Magnetic field strength} at point $\mathbf{r}$ and
time $t$ is
\begin{equation}
\mathbf{B}(\mathbf{r},t)\overset{^{def}}{=}\left.\left(\begin{array}{c}
\partial_{v_{z}}U_{y}^{\mathbf{r},t}\\
\partial_{v_{x}}U_{z}^{\mathbf{r},t}\\
\partial_{v_{y}}U_{x}^{\mathbf{r},t}
\end{array}\right)\right|_{\mathbf{v}=0}\label{eq:B-def}
\end{equation}

\noindent \medskip{}
Practically it means that one can determine the value of $\mathbf{B}(\mathbf{r},t)$,
with arbitrary precision, by means of measuring the accelerations
of a few test particles of velocity $\mathbf{v}^{e}\in\delta$.

Next we introduce the concepts of source densities:

\paragraph*{Definition~(D3)}

\begin{eqnarray}
\varrho\left(\mathbf{r},t\right) & \overset{^{def}}{=} & \nabla\cdot\mathbf{E}\left(\mathbf{r},t\right)\label{eq:ME1}\\
\mathbf{j}\left(\mathbf{r},t\right) & \overset{^{def}}{=} & c^{2}\nabla\times\mathbf{B}\left(\mathbf{r},t\right)-\partial_{t}\mathbf{E}\left(\mathbf{r},t\right)\label{eq:ME2}
\end{eqnarray}
are called \emph{active electric charge density} and \emph{active
electric current density,} respectively. \medskip{}

A simple consequence of the \emph{definitions} is that a continuity
equation holds for $\varrho$ and $\mathbf{j}$:
\begin{thm}
~
\begin{equation}
\partial_{t}\varrho\left(\mathbf{r},t\right)+\nabla\cdot\mathbf{j}\left(\mathbf{r},t\right)=0\label{eq:kontinuitas}
\end{equation}

\end{thm}
\R\label{meg:2}In our construction, the two Maxwell equations (\ref{eq:ME1})--(\ref{eq:ME2}),
are mere \emph{definitions} of the concepts of active electric charge
density and\emph{ }active electric current density. They do not contain
information whatsoever about how ``matter produces electromagnetic
field''. And it is not because $\varrho\left(\mathbf{r},t\right)$
and $\mathbf{j}\left(\mathbf{r},t\right)$ are, of course, ``unspecified
distributions'' in these ``general laws'', but because $\varrho\left(\mathbf{r},t\right)$
and $\mathbf{j}\left(\mathbf{r},t\right)$ cannot be specified prior
to or at least independently of the field strengths $\mathbf{E}(\mathbf{r},t)$
and $\mathbf{B}(\mathbf{r},t)$. Again, because $\varrho\left(\mathbf{r},t\right)$
and $\mathbf{j}\left(\mathbf{r},t\right)$ are just abbreviations,
standing for the expressions on the right hand sides of (\ref{eq:ME1})--(\ref{eq:ME2}).
In other words, any statement about the ``charge distribution''
will be a statement about $\nabla\cdot\mathbf{E}$, and any statement
about the ``current distribution'' will be a statement about $c^{2}\nabla\times\mathbf{B}-\partial_{t}\mathbf{E}$.

The minimal claim is that this is a possible coherent construction.
Though we must add: equations (\ref{eq:ME1})--(\ref{eq:ME2}) could
be seen as contingent physical laws about the relationship between
the charge and current distributions and the electromagnetic field,
only if we had an independent empirical definition of charge. However,
we do not see how such a definition is possible, without encountering
circularities. (Also see Remark~\ref{meg:3}) \hfill{} $\lrcorner$\medskip{}

\noindent The operational definitions of the field strengths and the
source densities are based on the kinematic properties of the test
particles. The following definition describes the concept of a charged
point-like particle, in general.

\paragraph*{Definition (D4)}

A particle $b$ is called \emph{charged point-particle} of \emph{specific
passive electric charge} $\pi^{b}$ and of \emph{active electric charge}
$\alpha^{b}$ if the following is true:
\begin{enumerate}
\item It satisfies the relativistic Lorentz equation,

\noindent 
\begin{eqnarray}
\gamma\left(\mathbf{v}^{b}\left(t\right)\right)\mathbf{a}^{b}(t) & = & \pi^{b}\left\{ \mathbf{E}\left(\mathbf{r}^{b}\left(t\right),t\right)+\mathbf{v}^{b}\left(t\right)\times\mathbf{B}\left(\mathbf{r}^{b}\left(t\right),t\right)\right.\nonumber \\
 &  & \left.-c^{-2}\mathbf{v}^{b}\left(t\right)\left(\mathbf{v}^{b}\left(t\right)\mathbf{\cdot E}\left(\mathbf{r}^{b}\left(t\right),t\right)\right)\right\} \label{eq:E1'}
\end{eqnarray}

\item \noindent If it is the only particle whose worldline intersects a
given space-time region $\Lambda$, then for all $(\mathbf{r},t)\in\Lambda$
the source densities are of the following form: 
\begin{eqnarray}
\varrho\left(\mathbf{r},t\right) & = & \alpha^{b}\delta\left(\mathbf{r}-\mathbf{r}^{b}\left(t\right)\right)\label{eq:(E3)1}\\
\mathbf{j}\left(\mathbf{r},t\right) & = & \alpha^{b}\delta\left(\mathbf{r}-\mathbf{r}^{b}\left(t\right)\right)\mathbf{v}^{b}\left(t\right)\label{eq:(E3)2}
\end{eqnarray}

\end{enumerate}
where $\mathbf{r}^{b}\left(t\right)$, $\mathbf{v}^{b}\left(t\right)$
and $\mathbf{a}^{b}\left(t\right)$ are the particle's position, velocity
and acceleration. The ratio $\mu^{b}\overset{^{def}}{=}\alpha^{b}/\pi^{b}$
is called the \emph{electric inertial rest mass} of the particle.

\R\label{meg:3} Of course, \eqref{eq:E1'} is equivalent to the
standard form of the Lorentz equation: 
\begin{equation}
\frac{d}{dt}\left(\gamma\left(\mathbf{v}\left(t\right)\right)\mathbf{v}\left(t\right)\right)=\pi\left\{ \mathbf{E}\left(\mathbf{r}\left(t\right),t\right)+\mathbf{v}\left(t\right)\times\mathbf{B}\left(\mathbf{r}\left(t\right),t\right)\right\} 
\end{equation}
with $\pi=q/m$ in the usual terminology, where $q$ is the passive
electric charge and $m$ is the inertial (rest) mass of the particle---that
is why we call $\pi$ \emph{specific} passive electric charge. Nevertheless,
it must be clear that for all charged point-particles we introduced
\emph{two independent}, empirically meaningful and experimentally
testable quantities: specific passive electric charge $\pi$ and active
electric charge $\alpha$. There is no universal law-like relationship
between these two quantities: the ratio between them varies from particle
to particle. In the traditional sense, this ratio is, however, nothing
but the particle's rest mass.

We must emphasize that the concept of mass so obtained, as defined
by only means of electrodynamic quantities, is essentially related
to ED, that is to say, to electromagnetic interaction. There seems
no way to give a consistent and non-circular operational definition
of inertial mass in general, independently of the context of a particular
type of physical interaction. Without entering here into the detailed
discussion of the problem, we only mention that, for example, Weyl's
commonly accepted definition (Jammer 2000, pp. 8--10) and all similar
definitions based on the conservation of momentum in particle collisions
suffer from the following difficulty. There is no ``collision''
as a purely ``mechanical'' process. During a collision the particles
are moving in a physical field---or fields---of interaction. Therefore:
1)~the system of particles, separately, cannot be regarded as a closed
system;\emph{ }2)~the inertial properties of the particles, in fact,
reveal themselves in the interactions with the field.\emph{ }Thus,
the concepts of inertial rest mass belonging to different interactions
differ from each other; whether they are equal (proportional) to each
other is a matter of contingent fact of nature.\emph{ }\hfill{} $\lrcorner$\medskip{}

\R\label{meg:4}The choice of the \emph{etalon} test particle is,
of course, a matter of convention, just as the definitions (D0)--(D4)
themselves. It is important to note that all these conventional factors
play a constitutive role in the fundamental concepts of ED (Reichenbach
1965). With these choices we not only make semantic\emph{ }conventions
determining the meanings of the terms, but also make a decision about
the body of concepts by means of which we grasp physical reality.\emph{
}There are a few things, however, that must be pointed out:
\begin{lyxlist}{00.00.0000}
\item [{(a)}] This kind of conventionality does not mean that the physical
quantities defined in (D0)--(D4) cannot describe \emph{objective}
features of physical reality. It only means that we make a decision
which objective features of reality we are dealing with. With another
body of conventions we have another body of physical concepts/physical
quantities and another body of empirical facts.
\item [{(b)}] \noindent On the other hand, it does not mean either that
our knowledge of the physical world would not be objective but a product
of our conventions. If two theories obtained by starting with two
different bodies of conventions are complete enough accounts of the
physical phenomena, then they describe the same reality, expressed
in terms of different physical quantities. Let us spell out an example:
Definition \eqref{eq:ME2} is entirely conventional---no objective
fact of the world determines the formula on the right hand side. Therefore,
we could make another choice, say,
\begin{equation}
\mathbf{j}_{\Theta}\left(\mathbf{r},t\right)\overset{^{def}}{=}\Theta^{2}\nabla\times\mathbf{B}\left(\mathbf{r},t\right)-\partial_{t}\mathbf{E}\left(\mathbf{r},t\right)\label{eq:tetasje}
\end{equation}
with some $\Theta\neq c$. At first sight, one might think that this
choice will alter the speed of electromagnetic waves. This is however
not the case. It will be an empirical fact \emph{about} $\mathbf{j}_{\Theta}\left(\mathbf{r},t\right)$
that if a particle $b$ is the only one whose worldline intersects
a given space-time region $\Lambda$, then for all $(\mathbf{r},t)\in\Lambda$
\begin{eqnarray}
\mathbf{j}_{\Theta}\left(\mathbf{r},t\right) & = & \alpha^{b}\delta\left(\mathbf{r}-\mathbf{r}^{b}\left(t\right)\right)\mathbf{v}^{b}\left(t\right)\nonumber \\
 &  & +\left(\Theta^{2}-c^{2}\right)\nabla\times\mathbf{B}(\mathbf{r},t)\label{eq:jeteta}
\end{eqnarray}
Now, consider a region where there is no particle. Taking into account
\eqref{eq:jeteta}, we have \eqref{eq:ME3}--\eqref{eq:ME4} and 
\begin{eqnarray}
\nabla\cdot\mathbf{E}(\mathbf{r},t) & = & 0\\
\Theta^{2}\nabla\times\mathbf{B}\left(\mathbf{r},t\right)-\partial_{t}\mathbf{E}\left(\mathbf{r},t\right) & = & \left(\Theta^{2}-c^{2}\right)\nabla\times\mathbf{B}(\mathbf{r},t)
\end{eqnarray}
which lead to the usual wave equation with propagation speed $c$.
(Of course, in this particular example, one of the possible choices,
namely $\Theta=c$, is distinguished by its simplicity. Note, however,
that simplicity is not an epistemologically interpretable notion.)\hfill{}
$\lrcorner$
\end{lyxlist}

\section{Empirical facts of electrodynamics}

Both ``empirical'' and ``fact'' are used in different senses.
Statements (E0)--(E4) below are universal generalizations, rather
than statements of particular observations. Nevertheless we call them
``empirical facts'', by which we simply mean that they are truths
which can be acquired by \emph{a posteriori} means. Normally, they
can be considered as laws obtained by inductive generalization; statements
the truths of which can be, in principle, confirmed empirically. 

On the other hand, in our context, it is not important how these statements
are empirically confirmed. (E0)--(E4) can be regarded as axioms of
the Maxwell--Lorentz theory in $K$. What is important for us is that
from these \emph{axioms}, in conjunction with the theoretical representations
of the measurement operations, there follow assertions about what
the moving observer in $K'$ observes. Section~\ref{sec:Observations-of-moving}
will be concerned with these consequences.

\paragraph*{(E0)}

There exist many enough test particles and we can settle them into
all required positions and velocities.

\medskip{}

\noindent Consequently, (D1)--(D4) are sound definitions. From observations
about $\mathbf{E}$, $\mathbf{B}$ and the charged point-particles,
we have further empirical facts:

\paragraph*{(E1)}

In all situations, the electric and magnetic field strengths satisfy
the following two Maxwell equations:

\noindent 
\begin{eqnarray}
\nabla\cdot\mathbf{B}\left(\mathbf{r},t\right) & = & 0\label{eq:ME3}\\
\nabla\times\mathbf{E}\left(\mathbf{r},t\right)+\partial_{t}\mathbf{B}\left(\mathbf{r},t\right) & = & 0\label{eq:ME4}
\end{eqnarray}

\paragraph*{(E2)}

Each particle is a charged point-particle, satisfying (D4) with some
specific passive electric charge $\pi$ and active electric charge
$\alpha$. This is also true for the test particles, with---as follows
from the definitions---specific passive electric charge $\pi=1$.%
\footnote{We take it true that the relativistic Lorentz equation is empirically
confirmed. (Cf. Huang 1993) %
}

\paragraph*{(E3)}

If $b_{1}$, $b_{2}$,..., $b_{n}$ are the only particles whose worldlines
intersect a given space-time region $\Lambda$, then for all $(\mathbf{r},t)\in\Lambda$
the source densities are:
\begin{eqnarray}
\varrho\left(\mathbf{r},t\right) & = & \sum_{i=1}^{n}\alpha^{b_{i}}\delta\left(\mathbf{r}-\mathbf{r}^{b_{i}}\left(t\right)\right)\label{eq:add1}\\
\mathbf{j}\left(\mathbf{r},t\right) & = & \sum_{i=1}^{n}\alpha^{b_{i}}\delta\left(\mathbf{r}-\mathbf{r}^{b_{i}}\left(t\right)\right)\mathbf{v}^{b_{i}}\left(t\right)\label{eq:add2}
\end{eqnarray}
\medskip{}

Putting facts (E1)--(E3) together, we have the coupled Maxwell--Lorentz
equations:
\begin{eqnarray}
\nabla\cdot\mathbf{E}\left(\mathbf{r},t\right) & = & \sum_{i=1}^{n}\alpha^{b_{i}}\delta\left(\mathbf{r}-\mathbf{r}^{b_{i}}\left(t\right)\right)\label{eq:MLE1}\\
c^{2}\nabla\times\mathbf{B}\left(\mathbf{r},t\right)-\partial_{t}\mathbf{E}\left(\mathbf{r},t\right) & = & \sum_{i=1}^{n}\alpha^{b_{i}}\delta\left(\mathbf{r}-\mathbf{r}^{b_{i}}\left(t\right)\right)\mathbf{v}^{b_{i}}\left(t\right)\label{eq:MLE2}\\
\nabla\cdot\mathbf{B}\left(\mathbf{r},t\right) & = & 0\label{eq:MLE3}\\
\nabla\times\mathbf{E}\left(\mathbf{r},t\right)+\partial_{t}\mathbf{B}\left(\mathbf{r},t\right) & = & 0\label{eq:MLE4}\\
\gamma\left(\mathbf{v}^{b_{i}}\left(t\right)\right)\mathbf{a}^{b_{i}}(t) & = & \pi^{b_{i}}\left\{ \mathbf{E}\left(\mathbf{r}^{b_{i}}\left(t\right),t\right)+\mathbf{v}^{b_{i}}\left(t\right)\times\mathbf{B}\left(\mathbf{r}^{b_{i}}\left(t\right),t\right)\right.\,\,\,\nonumber \\
 &  & \left.-c^{-2}\mathbf{v}^{b_{i}}\left(t\right)\left(\mathbf{v}^{b_{i}}\left(t\right)\mathbf{\cdot E}\left(\mathbf{r}^{b_{i}}\left(t\right),t\right)\right)\right\} \label{eq:MLE5}\\
 &  & \,\,\,\,\,\,\,\,\,\,\,\,\,\,\,\,\,\,\,\,\,\,\,\,\,\,\,\,\,(i=1,2,\ldots n)\nonumber 
\end{eqnarray}

\noindent These are the fundamental equations of ED, describing an
interacting system of $n$ particles and the electromagnetic field. 

\R\label{meg:3a} Without entering into the details of the problem
of classical charged particles (Frisch 2005; Rohrlich 2007; Muller
2007), it must be noted that the Maxwell--Lorentz equations (\ref{eq:MLE1})--(\ref{eq:MLE5}),
exactly in this form, have \emph{no} solution. The reason is the following.
In the Lorentz equation of motion (\ref{eq:E1'}), a small but extended
particle can be described with a good approximation by one single
specific passive electric charge $\pi^{b}$ and one single trajectory
$\mathbf{r}^{b}\left(t\right)$. In contrast, however, a similar ``idealization''
in the source densities (\ref{eq:(E3)1})--(\ref{eq:(E3)2}) leads
to singularities; the field is singular at precisely the points where
the coupling happens: on the trajectory of the particle. 

The generally accepted answer to this problem is that (\ref{eq:(E3)1})--(\ref{eq:(E3)2})
should not be taken literally. Due to the inner structure of the particle,
the real source densities are some ``smoothed out'' Dirac deltas.
Instead of (\ref{eq:(E3)1})--(\ref{eq:(E3)2}), therefore, we have
some more general equations
\begin{eqnarray}
\left[\varrho(\mathbf{r},t)\right] & = & \mathcal{R}^{b}\left[\mathbf{r}^{b}(t)\right]\label{eq:struktura1}\\
\left[\mathbf{j}(\mathbf{r},t)\right] & = & \mathcal{J}^{b}\left[\mathbf{r}^{b}(t)\right]\label{eq:struktura2}
\end{eqnarray}
where $\mathcal{R}^{b}$ and $\mathcal{J}^{b}$ are, generally non-linear,
operators providing functional relationships between the particle's
trajectory $\left[\mathbf{r}^{b}(t)\right]$ and the source density
functions $\left[\varrho(\mathbf{r},t)\right]$ and $\left[\mathbf{j}(\mathbf{r},t)\right]$.
(Notice that (\ref{eq:(E3)1})--(\ref{eq:(E3)2}) serve as example
of such equations.) The concrete forms of equations (\ref{eq:struktura1})--(\ref{eq:struktura2})
are determined by the physical laws of the internal world of the particle---which
are, supposedly, outside of the scope of ED. At this level of generality,
the only thing we can say is that, for a ``point-like'' (localized)
particle, equations (\ref{eq:struktura1})--(\ref{eq:struktura2})
must be something very close to---but not identical with---equations
(\ref{eq:(E3)1})--(\ref{eq:(E3)2}). With this explanation, for the
sake of simplicity we leave the Dirac deltas in the equations. Also,
in some of our statements and calculations the Dirac deltas are essentially
used; for example, (E3) and, partly, Theorem~\ref{thm:-E2'} and
\ref{thm:superpos'} would not be true without the exact point-like
source densities (\ref{eq:(E3)1})--(\ref{eq:(E3)2}). But a little
reflection shows that the statements in question remain approximately
true if the particles are approximately point-like, that is, if equations
(\ref{eq:struktura1})--(\ref{eq:struktura2}) are close enough to
equations (\ref{eq:(E3)1})--(\ref{eq:(E3)2}). To be noted that what
is actually essential in (\ref{eq:(E3)1})--(\ref{eq:(E3)2}) is not
the point-likeness of the particle, but its stability: no matter how
the system moves, it remains a localized object. \hfill{} $\lrcorner$

\section{Operational definitions of electrodynamic quantities in $K'$}

So far we have only considered ED in a single frame of reference $K$.
Now we turn to the question of how a moving observer describes the
same phenomena in $K'$. The observed phenomena are the same, but
the measuring equipments by means of which the phenomena are observed
are not entirely the same; instead of being at rest in $K$, they
are co-moving with\emph{ $K'$}.

Accordingly, we will repeat the operational definitions (D0)--(D4)
with the following differences:
\begin{enumerate}
\item The ``rest test particles'' will be at rest relative to reference
frame $K'$, that is, \emph{in motion with velocity $\mathbf{V}$}
relative to $K$.
\item The measuring equipments by means of which the kinematic quantities
are ascertained---say, the measuring rods and clocks---will be at
rest relative to $K'$, that is, \emph{in motion with velocity $\mathbf{V}$}
relative to $K$. In other words, the kinematic quantities $t,\mathbf{r},\mathbf{v},\mathbf{a}$
in definitions (D0)--(D4) will be \emph{replaced with}---not expressed
in terms of--- $t',\mathbf{r}',\mathbf{v}',\mathbf{a}'$.
\end{enumerate}

\paragraph*{Definition~(D0')}

Particle $e$ is called \emph{(test particle)'} if for all $\mathbf{r}'$
and $t'$
\begin{equation}
\mathbf{v}'^{e}\left(t'\right)\biggl|_{\mathbf{r}'^{e}\left(t'\right)=\mathbf{r}'}=\mathbf{v}'^{etalon}\left(t'\right)\biggl|_{\mathbf{r}'^{etalon}\left(t'\right)=\mathbf{r}'}
\end{equation}
implies
\begin{equation}
\mathbf{a}'^{e}\left(t'\right)\biggl|_{\mathbf{r}'^{e}\left(t'\right)=\mathbf{r}'}=\mathbf{a}'^{etalon}\left(t'\right)\biggl|_{\mathbf{r}'^{etalon}\left(t'\right)=\mathbf{r}'}
\end{equation}
\medskip{}

\noindent A (test particle)' $e$ moving with velocity $\mathbf{V}$
relative to $K$ is at rest relative to $K'$, that is, $\mathbf{v}'^{e}=0$.
Accordingly:

\paragraph*{Definition~(D1')}

\noindent \emph{(Electric field strength)'} at point $\mathbf{r}'$
and time $t'$ is defined as the acceleration of an arbitrary (test
particle)' $e$, such that $\mathbf{r}'^{e}(t)=\mathbf{r}'$ and $\mathbf{v}'^{e}(t')=0$:
\begin{equation}
\mathbf{E}'\left(\mathbf{r}',t'\right)\overset{^{def}}{=}\left.\mathbf{a}'^{e}(t')\right|_{\mathbf{r}'^{e}(t')=\mathbf{r}';\,\mathbf{v}'^{e}(t')=0}\label{eq:E-def-vesszo}
\end{equation}
\medskip{}

\noindent Similarly, (magnetic field strength)' is defined by means
of how the acceleration $\mathbf{a}'^{e}$ of a rest (test particle)'---rest,
of course, relative to $K'$---changes with a small perturbation of
its state of motion, that is, if an infinitesimally small velocity
$\mathbf{v}'^{e}$ is imparted to the particle. Just as in (D2), let
$\delta'\subset\mathbb{R}^{3}$ be an arbitrary infinitesimal neighborhood
of $0\in\mathbb{R}^{3}$. We define the following function:
\begin{eqnarray}
\mathbf{U}'^{\mathbf{r}',t'} & : & \mathbb{R}^{3}\supset\delta'\rightarrow\mathbb{R}^{3}\nonumber \\
 &  & \mathbf{U}'^{\mathbf{r}',t'}(\mathbf{v}')\overset{^{def}}{=}\left.\mathbf{a}'^{e}(t')\right|_{\mathbf{r}'^{e}(t')=\mathbf{r}';\,\mathbf{v}'^{e}(t')=\mathbf{v}'}\label{eq:U-def-vesszo}
\end{eqnarray}

\paragraph*{Definition~(D2')}

\noindent \emph{(Magnetic field strength)'} at point $\mathbf{r}'$
and time $t'$ is

\begin{equation}
\mathbf{B}'(\mathbf{r}',t')\overset{^{def}}{=}\left.\left(\begin{array}{l}
\partial_{v'_{z}}U'{}_{y}^{\mathbf{r}',t'}\\
\partial_{v'_{x}}U'{}_{z}^{\mathbf{r}',t'}\\
\partial_{v'_{y}}U'{}_{x}^{\mathbf{r}',t'}
\end{array}\right)\right|_{\mathbf{v}'=0}\label{eq:B-def-vesszo}
\end{equation}

\paragraph*{Definition~(D3')}

\begin{eqnarray}
\varrho'\left(\mathbf{r}',t'\right) & \overset{^{def}}{=} & \nabla\cdot\mathbf{E}'\left(\mathbf{r}',t'\right)\label{eq:ME1'}\\
\mathbf{j}'\left(\mathbf{r}',t'\right) & \overset{^{def}}{=} & c^{2}\nabla\times\mathbf{B}'\left(\mathbf{r}',t'\right)-\partial_{t'}\mathbf{E}'\left(\mathbf{r}',t'\right)\label{eq:ME2'}
\end{eqnarray}
are called \emph{(active electric charge density)'} and \emph{(active
electric current density)',} respectively. \medskip{}

\noindent Of course, we have:
\begin{thm}
~
\begin{equation}
\partial_{t'}\varrho'\left(\mathbf{r}',t'\right)+\nabla\cdot\mathbf{j}'\left(\mathbf{r}',t'\right)=0\label{eq:kontinuitas'}
\end{equation}

\end{thm}

\paragraph*{Definition (D4')}

A particle is called \emph{(charged point-particle)'} of \emph{(specific
passive electric charge)'} $\pi'^{b}$ and of \emph{(active electric
charge)'} $\alpha'^{b}$ if the following is true:
\begin{enumerate}
\item It satisfies the relativistic Lorentz equation,

\noindent 
\begin{eqnarray}
\gamma\left(\mathbf{v}'^{b}\left(t'\right)\right)\mathbf{a}'^{b}(t') & = & \pi'^{b}\left\{ \mathbf{E}'\left(\mathbf{r}'^{b}\left(t'\right),t'\right)+\mathbf{v}'^{b}\left(t'\right)\times\mathbf{B}'\left(\mathbf{r}'^{b}\left(t'\right),t'\right)\right.\nonumber \\
 &  & \left.-c^{-2}\mathbf{v}'^{b}\left(t'\right)\left(\mathbf{v}'^{b}\left(t'\right)\mathbf{\cdot E}'\left(\mathbf{r}'^{b}\left(t'\right),t'\right)\right)\right\} \label{eq:E1'-vesszo}
\end{eqnarray}

\item \noindent If it is the only particle whose worldline intersects a
given space-time region $\Lambda'$, then for all $(\mathbf{r}',t')\in\Lambda'$
the (source densities)' are of the following form: 
\begin{eqnarray}
\varrho'\left(\mathbf{r}',t'\right) & = & \alpha'^{b}\delta\left(\mathbf{r}'-\mathbf{r}'^{b}\left(t'\right)\right)\label{eq:(E3)1-vesszo}\\
\mathbf{j}'\left(\mathbf{r}',t'\right) & = & \alpha'^{b}\delta\left(\mathbf{r}'-\mathbf{r}'^{b}\left(t'\right)\right)\mathbf{v}'^{b}\left(t'\right)\label{eq:(E3)2-vesszo}
\end{eqnarray}

\end{enumerate}
where $\mathbf{r}'^{b}\left(t'\right)$, $\mathbf{v}'^{b}\left(t'\right)$
and $\mathbf{a}'^{b}\left(t'\right)$ is the particle's position,
velocity and acceleration in $K'$. The ratio $\mu'^{b}\overset{^{def}}{=}\alpha'^{b}/\pi'^{b}$
is called the \emph{(electric inertial rest mass)'} of the particle.\medskip{}

\R\label{meg:5}It is worthwhile to make a few remarks about some
epistemological issues:
\begin{lyxlist}{00.00.0000}
\item [{(a)}] The physical quantities defined in (D1)--(D4) \emph{differ}
from the physical quantities defined in (D1')--(D4'), simply because
the physical situation in which a test particle is at rest relative
to $K$ differs from the one in which it is co-moving with $K'$ with
velocity $\mathbf{V}$ relative to $K$; and, as we know \emph{from
the laws of ED in $K$}, this difference really matters.

Someone might object that if this is so then any two instances of
the same measurement must be regarded as measurements of different
physical quantities. For, if the difference in the test particle's
velocity is enough reason to say that the two operations determine
two different quantities, then, by the same token, two operations
must be regarded as different operations---and the corresponding quantities
as different physical quantities---if the test particle is at different
points of space, or the operations simply happen at different moments
of time. And this consequence, the objection goes, seems to be absurd:
if it were true, then science would not be possible, because we would
not have the power to make law-like assertions at all; therefore we
must admit that empiricism fails to explain how natural laws are possible,
and, as many argue, science cannot do without metaphysical pre-assumptions.

Our response to such an objections is the following. First, concerning
the general epistemological issue, we believe, nothing disastrous
follows from admitting that two phenomena observed at different place
or at different time \emph{are} distinct. And if they are stated as
instances of the same phenomenon, this statement is not a logical
or metaphysical necessity---derived from some logical/metaphysical
pre-assumptions---but an ordinary scientific hypothesis obtained by
induction and confirmed or disconfirmed together with the \emph{whole}
scientific theory. In fact, this is precisely the case with respect
to the definitions of the fundamental electrodynamic quantities. For
example, definition (D1) is in fact a family of definitions each belonging
to a particular situation individuated by the space-time locus $(\mathbf{r},t)$.\\
Second, the question of operational definitions of electrodynamic
quantities first of all emerges not from an epistemological context,
but from the context of a purely theoretical problem: what do the
laws of physics in $K$ say about question (Q)? In the next section,
all the results of the measurement operations defined in (D1')--(D4')
will be predicted from the laws of ED in $K$. And, ED itself says
that some differences in the conditions are relevant from the point
of view of the measured accelerations of the test particles, some
others are not; some of the originally distinct quantities are contingently
equal, some others not.

\item [{(b)}] From a mathematical point of view, both (D0)--(D4) and (D0')--(D4')
are definitions. However, while the choice of the \emph{etalon} test
particle and definitions (D0)--(D4) are entirely \emph{conventional},
there is no additional conventionality in (D0')--(D4'). The way in
which we define the electrodynamic quantities in inertial frame $K'$
automatically follows from (D0)--(D4) and from the question (Q) we
would like to answer; since the question is about the ``quantities
obtained by the same operational procedures with the same measuring
equipments when they are co-moving with\emph{ }$K'$''.\emph{ }
\item [{(c)}] In fact, one of the constituents of the concepts defined
in $K'$ is not determined by the operational definitions in $K$.
Namely, the notion of ``the same operational procedures with the
same measuring equipments when they are co-moving with\emph{ $K'$}''.
This is however not an additional freedom of conventionality, but
a simple vagueness in our physical theories in $K$: the vagueness
of the general concept of ``the same system in the same situation,
except that it is, as a whole, in a collective motion with velocity
$\mathbf{V}$ relative to $K$, that is, co-moving with reference
frame $K'$'' (Szabó~2004; Gömöri and Szabó 2011). In any event,
in our case, the notion of the only moving measuring device, that
is, the notion of ``a test particle at rest relative to $K'$''
is quite clear.\hfill{} $\lrcorner$
\end{lyxlist}

\section{Observations of moving observer\label{sec:Observations-of-moving}}

Now we have another collection of operationally defined notions, $\mathbf{E}',\mathbf{B}',$$\varrho',\mathbf{j}'$,
the concept of (charged point-particle)' defined in the primed terms,
and its properties $\pi',\alpha'$ and $\mu'$. Normally, one should
investigate these quantities experimentally and collect new empirical
facts about both the relationships between the primed quantities and
about the relationships between the primed quantities and the ones
defined in (D1)--(D4). In contrast, we will continue our analysis
in another way; following the ``Lorentzian pedagogy'', we will determine
from the laws of physics in $K$ what an observer co-moving with $K'$
should observe. In fact, with this method, we will answer our question
(Q), on the basis of the laws of ED in one single frame of reference.
We will also see whether the basic equations \eqref{eq:MLE1}--\eqref{eq:MLE5}
are covariant against these transformations.

Throughout the theorems below, it is important that when we compare,
for example, $\mathbf{E}\left(\mathbf{r},t\right)$ with $\mathbf{E}'(\mathbf{r}',t')$,
we compare the values of the fields \emph{in one and the same event},
that is, we compare $\mathbf{E}\left(\mathbf{r}(A),t(A)\right)$ with
$\mathbf{E}'\left(\mathbf{r}'(A),t'(A)\right)$. For the sake of brevity,
however, we omit the indication of this fact.

The first theorem trivially follows from the fact that the Lorentz
transformations of the kinematic quantities are one-to-one: 
\begin{thm}
\label{thm:test=00003Dtest'}A particle is a (test particle)' if and
only if it is a test particle.
\end{thm}
\noindent Consequently, we have many enough (test particles)' for
definitions (D1')--(D4'); and each is a charged point-particle satisfying
the Lorentz equation \eqref{eq:E1'} with specific passive electric
charge $\pi=1$. 
\begin{thm}
~\label{thm:E'E}
\begin{eqnarray}
E'_{x} & = & E_{x}\label{eq:E'x}\\
E'_{y} & = & \gamma\left(E_{y}-VB_{z}\right)\label{eq:E'y}\\
E'_{z} & = & \gamma\left(E_{z}+VB_{y}\right)\label{eq:E'z}
\end{eqnarray}
\end{thm}
\begin{proof}
\noindent When the (test particle)' is at rest relative to $K'$,
it is moving with velocity $\mathbf{v}^{e}=\left(V,0,0\right)$ relative
to $K$. From \eqref{eq:E1'} (with $\pi=1$) we have
\begin{eqnarray}
a_{x}^{e} & = & \gamma^{-3}E_{x}\label{eq:TTgyors1}\\
a_{y}^{e} & = & \gamma^{-1}\left(E_{y}-VB_{z}\right)\\
a_{z}^{e} & = & \gamma^{-1}\left(E_{z}+VB_{y}\right)\label{eq:TTgyors3}
\end{eqnarray}
Applying \eqref{eq:gyorsulas1}--\eqref{eq:gyorsulas3}, we can calculate
the acceleration $\mathbf{a}'^{e}$ in $K'$, and, accordingly, we
find
\begin{eqnarray}
E'_{x} & = & a{}_{x}^{\prime e}=\gamma^{3}a_{x}^{e}=E_{x}\\
E'_{y} & = & a{}_{y}^{\prime e}=\gamma^{2}a_{y}^{e}=\gamma\left(E_{y}-VB_{z}\right)\\
E'_{z} & = & a{}_{z}^{\prime e}=\gamma^{2}a_{z}^{e}=\gamma\left(E_{z}+VB_{y}\right)
\end{eqnarray}
 \end{proof}
\begin{thm}
\label{thm:B'B}
\begin{eqnarray}
B'_{x} & = & B{}_{x}\label{eq:Bx'}\\
B'_{y} & = & \gamma\left(B_{y}+c^{-2}VE_{z}\right)\\
B'_{z} & = & \gamma\left(B_{z}-c^{-2}VE_{y}\right)\label{eq:Bz'}
\end{eqnarray}
\end{thm}
\begin{proof}
\noindent Consider for instance $B'_{x}$. By definition, 
\begin{equation}
B'_{x}=\left.\partial_{v'_{z}}U'{}_{y}^{\mathbf{r}',t'}\right|_{\mathbf{v}'=0}\label{eq:Bxdef}
\end{equation}
According to \eqref{eq:U-def-vesszo}, the value of $U'{}_{y}^{\mathbf{r}',t'}(\mathbf{v}')$
is equal to 
\begin{equation}
\left.a'{}_{y}^{e}\right|_{\mathbf{r}'^{e}(t')=\mathbf{r}';\,\mathbf{v}'^{e}(t')=\mathbf{v}'}
\end{equation}
that is, the $y$-component of the acceleration of a (test particle)'
$e$ in a situation in which $\mathbf{r}'^{e}(t')=\mathbf{r}'$ and
$\mathbf{v}'^{e}(t')=\mathbf{v}'$. Accordingly, in order to determine
the partial derivative \eqref{eq:Bxdef} we have to determine 
\begin{equation}
\left.\frac{d}{dw}\right|_{w=0}\left(\left.a'{}_{y}^{e}\right|_{\mathbf{r}'^{e}(t')=\mathbf{r}';\,\mathbf{v}'^{e}(t')=\left(0,0,w\right)}\right)
\end{equation}
Now, according to \eqref{eq:sebesseg1}, condition $\mathbf{v}'^{e}=\left(0,0,w\right)$
corresponds to 
\begin{equation}
\mathbf{v}^{e}=\left(V,0,\gamma^{-1}w\right)
\end{equation}
Substituting this velocity into \eqref{eq:E1'}, we have:
\begin{equation}
a_{y}^{e}=\sqrt{1-\frac{V^{2}+w{}^{2}\gamma^{-2}}{c^{2}}}\left(E_{y}+w\gamma^{-1}B_{x}-VB_{z}\right)\label{eq:hivatkozott-a}
\end{equation}
Applying \eqref{eq:gyorsulas4}, one finds:
\begin{eqnarray}
a{}_{y}^{\prime e} & = & \gamma^{2}a_{y}^{e}=\gamma^{2}\sqrt{1-\frac{V^{2}+w{}^{2}\gamma^{-2}}{c^{2}}}\left(E_{y}+w\gamma^{-1}B_{x}-VB_{z}\right)\nonumber \\
 & = & \frac{\gamma}{\gamma(w)}\left(E_{y}+w\gamma^{-1}B_{x}-VB_{z}\right)\label{eq:gyorsulasB}
\end{eqnarray}
Differentiating with respect to $w$ at $w=0$, we obtain
\begin{equation}
B'_{x}=B{}_{x}
\end{equation}
The other components can be obtained in the same way.\end{proof}
\begin{thm}
\label{thm:aramtrafo}~
\begin{eqnarray}
\varrho' & = & \gamma\left(\varrho-c^{-2}Vj_{x}\right)\label{eq:surusegtarfo1}\\
j'_{x} & = & \gamma\left(j_{x}-V\varrho\right)\\
j'_{y} & = & j_{y}\\
j'_{z} & = & j_{z}\label{eq:surusegtrafo4}
\end{eqnarray}
\end{thm}
\begin{proof}
In \eqref{eq:ME1'} and \eqref{eq:ME2'}, substituting $\mathbf{E}'$
and $\mathbf{B}'$ with the right-hand-sides of \eqref{eq:E'x}--\eqref{eq:E'z}
and \eqref{eq:Bx'}--\eqref{eq:Bz'}, $\mathbf{r}$ and $t$ with
the inverse of \eqref{eq:LT1}--\eqref{eq:LT4}, then differentiating
the composite function and taking into account \eqref{eq:ME1}--\eqref{eq:ME2},
we get \eqref{eq:surusegtarfo1}--\eqref{eq:surusegtrafo4}.\end{proof}
\begin{thm}
\label{thm:-E2'}A particle $b$ is charged point-particle of specific
passive electric charge\emph{ }$\pi{}^{b}$ and of active electric\emph{
}charge $\alpha{}^{b}$ if and only if it is a (charged point-particle)'
of (specific passive electric charge)'\emph{ }$\pi'^{b}$ and of (active
electric\emph{ }charge)' $\alpha'{}^{b}$, such that $\pi'^{b}=\pi{}^{b}$
and $\alpha'{}^{b}=\alpha^{b}$. \end{thm}
\begin{proof}
First we prove \eqref{eq:E1'-vesszo}. For the sake of simplicity,
we will verify this in case of $\mathbf{v}{}^{\prime b}=\left(0,0,w\right)$.
We can use \eqref{eq:hivatkozott-a}: 
\begin{eqnarray}
a_{y}^{b} & = & \pi^{b}\sqrt{1-\frac{V^{2}+w{}^{2}\gamma^{-2}}{c^{2}}}\left(E_{y}+w\gamma^{-1}B_{x}-VB_{z}\right)
\end{eqnarray}
From \eqref{eq:gyorsulas4}, \eqref{eq:E'y}, \eqref{eq:Bx'}, and
\eqref{eq:Bz'} we have 
\begin{eqnarray}
a{}_{y}^{\prime b} & = & \pi^{b}\gamma(w)^{-1}\left(E'_{y}+wB'_{x}\right)\nonumber \\
 & = & \left[\pi^{b}\gamma\left(\mathbf{v}'^{b}\right)^{-1}\left(\mathbf{E}'-c^{-2}v'^{b}\left(\mathbf{v}'^{b}\mathbf{\cdot E}'\right)+\mathbf{v}'^{b}\times\mathbf{B}'\right)\right]_{y}\Biggl|_{\mathbf{v'}^{b}=\left(0,0,w\right)}
\end{eqnarray}
Similarly, 
\begin{eqnarray}
a{}_{x}^{\prime b} & = & \pi^{b}\gamma(w)^{-1}\left(E'_{x}-wB'_{y}\right)\nonumber \\
 & = & \left[\pi^{b}\gamma\left(\mathbf{v}'^{b}\right)^{-1}\left(\mathbf{E}'-c^{-2}\mathbf{v}'^{b}\left(\mathbf{v}'^{b}\mathbf{\cdot E}'\right)+\mathbf{v}'^{b}\times\mathbf{B}'\right)\right]_{x}\Biggl|_{\mathbf{v'}^{b}=\left(0,0,w\right)}\\
a{}_{z}^{\prime b} & = & \pi^{b}\gamma(w)^{-3}E'_{z}\nonumber \\
 & = & \left[\pi^{b}\gamma\left(\mathbf{v}'^{b}\right)^{-1}\left(\mathbf{E}'-c^{-2}\mathbf{v}'^{b}\left(\mathbf{v}'^{b}\mathbf{\cdot E}'\right)+\mathbf{v}'^{b}\times\mathbf{B}'\right)\right]_{z}\Biggl|_{\mathbf{v'}^{b}=\left(0,0,w\right)}
\end{eqnarray}
That is, \eqref{eq:E1'-vesszo} is satisfied, indeed.

In the second part, we show that \eqref{eq:(E3)1-vesszo}--\eqref{eq:(E3)2-vesszo}
are nothing but \eqref{eq:(E3)1}--\eqref{eq:(E3)2} expressed in
terms of $\mathbf{r}',t',\varrho'$ and $\mathbf{j}'$, with $\alpha^{'b}=\alpha^{b}$. 

It will be demonstrated for a particle of trajectory $\mathbf{r}'^{b}\left(t'\right)=\left(wt',0,0\right)$.
Applying \eqref{eq:sebesseg3}, \eqref{eq:(E3)1}--\eqref{eq:(E3)2}
have the following forms:
\begin{eqnarray}
\varrho{}\left(\mathbf{r},t\right) & = & \alpha^{b}\delta\left(x-\beta t\right)\delta\left(y\right)\delta\left(z\right)\\
\mathbf{j}{}\left(\mathbf{r},t\right) & = & \alpha^{b}\delta\left(x-\beta t\right)\delta\left(y\right)\delta\left(z\right)\left(\begin{array}{c}
\beta\\
0\\
0
\end{array}\right)
\end{eqnarray}
where $\beta=\frac{w+V}{1+c^{-2}wV}$. $\mathbf{r},t,\varrho$ and
$\mathbf{j}$ can be expressed with the primed quantities by applying
the inverse of \eqref{eq:LT1}--\eqref{eq:LT4} and \eqref{eq:surusegtarfo1}--\eqref{eq:surusegtrafo4}:
\begin{eqnarray}
\gamma\left(\varrho'{}\left(\mathbf{r}',t'\right)+c^{-2}Vj'{}_{x}\left(\mathbf{r}',t'\right)\right) & = & \alpha^{b}\delta\left(\gamma\left(x'+Vt'-\beta\left(t'+c^{-2}Vx'\right)\right)\right)\nonumber \\
 &  & \times\,\delta\left(y'\right)\delta\left(z'\right)\\
\gamma\left(j'{}_{x}\left(\mathbf{r}',t'\right)+V\varrho'{}\left(\mathbf{r}',t'\right)\right) & = & \alpha^{b}\delta\left(\gamma\left(x'+Vt'-\beta\left(t'+c^{-2}Vx'\right)\right)\right)\nonumber \\
 &  & \times\,\delta\left(y'\right)\delta\left(z'\right)\beta\\
j'{}_{y}\left(\mathbf{r}',t'\right) & = & 0\\
j'{}_{z}\left(\mathbf{r}',t'\right) & = & 0
\end{eqnarray}
One can solve this system of equations for $\varrho'$ and $j'_{x}$:
\begin{eqnarray}
\varrho'\left(\mathbf{r}',t'\right) & = & \alpha^{b}\delta\left(x'-wt'\right)\delta\left(y'\right)\delta\left(z'\right)\label{eq:suruseg-trafo_reszecske1}\\
\mathbf{j}'\left(\mathbf{r}',t'\right) & = & \alpha^{b}\delta\left(x'-wt'\right)\delta\left(y'\right)\delta\left(z'\right)\left(\begin{array}{c}
w\\
0\\
0
\end{array}\right)\label{eq:suruseg-trafo_reszecske2}
\end{eqnarray}
\end{proof}
\begin{thm}
~
\begin{eqnarray}
\nabla\cdot\mathbf{B}'\left(\mathbf{r}',t'\right) & = & 0\label{eq:ME3'}\\
\nabla\times\mathbf{E}'\left(\mathbf{r}',t'\right)+\partial_{t'}\mathbf{B}'\left(\mathbf{r}',t'\right) & = & 0\label{eq:ME4'}
\end{eqnarray}
\end{thm}
\begin{proof}
Expressing \eqref{eq:ME3}--\eqref{eq:ME4} in terms of $\mathbf{r}',t',\mathbf{E}'$
and $\mathbf{B}'$ by means of \eqref{eq:LT1}--\eqref{eq:LT4}, \eqref{eq:E'x}--\eqref{eq:E'z}
and \eqref{eq:Bx'}--\eqref{eq:Bz'}, we have
\begin{eqnarray}
\nabla\cdot\mathbf{B}'-c^{-2}V\left(\nabla\times\mathbf{E}'+\partial_{t'}\mathbf{B}'\right)_{x} & = & 0\\
\left(\nabla\times\mathbf{E}'+\partial_{t'}\mathbf{B}'\right)_{x}-V\nabla\cdot\mathbf{B}' & = & 0\\
\left(\nabla\times\mathbf{E}'+\partial_{t'}\mathbf{B}'\right)_{y} & = & 0\\
\left(\nabla\times\mathbf{E}'+\partial_{t'}\mathbf{B}'\right)_{z} & = & 0
\end{eqnarray}
which is equivalent to \eqref{eq:ME3'}--\eqref{eq:ME4'}.\end{proof}
\begin{thm}
\label{thm:superpos'}If $b_{1}$, $b_{2}$,..., $b_{n}$ are the
only particles whose worldlines intersect a given space-time region
$\Lambda'$, then for all $\left(\mathbf{r}',t'\right)\in\Lambda'$
the (source densities)' are:
\begin{eqnarray}
\varrho'\left(\mathbf{r}',t'\right) & = & \sum_{i=1}^{n}\alpha^{b_{i}}\delta\left(\mathbf{r}'-\mathbf{r}'^{b_{i}}\left(t'\right)\right)\label{eq:add1'}\\
\mathbf{j}'\left(\mathbf{r}',t'\right) & = & \sum_{i=1}^{n}\alpha^{b_{i}}\delta\left(\mathbf{r}'-\mathbf{r}'^{b_{i}}\left(t'\right)\right)\mathbf{v}'^{b_{i}}\left(t'\right)\label{eq:add2'}
\end{eqnarray}
\end{thm}
\begin{proof}
Due to Theorem~\ref{thm:-E2'}, each (charged point-particle)' is
a charged point-particle with $\alpha^{'b}=\alpha^{b}$. Therefore,
we only need to prove that equations \eqref{eq:add1'}--\eqref{eq:add2'}
amount to \eqref{eq:add1}--\eqref{eq:add2} expressed in the primed
variables. On the left hand side of \eqref{eq:add1}--\eqref{eq:add2},
$\varrho$ and $\mathbf{j}$ can be expressed by means ?of \eqref{eq:surusegtarfo1}--\eqref{eq:surusegtrafo4};
on the right hand side, we take $\alpha^{'b}=\alpha^{b}$, and apply
the inverse of \eqref{eq:LT1}--\eqref{eq:LT4}, just as in the derivation
of \eqref{eq:suruseg-trafo_reszecske1}--\eqref{eq:suruseg-trafo_reszecske2}.
From the above, we obtain: 

\begin{eqnarray}
\varrho'\left(\mathbf{r}',t'\right)+c^{-2}Vj'{}_{x}\left(\mathbf{r}',t'\right) & = & \sum_{i=1}^{n}\alpha^{b_{i}}\delta\left(\mathbf{r}'-\mathbf{r}'^{b_{i}}\left(t'\right)\right)\nonumber \\
 &  & +c^{-2}V\sum_{i=1}^{n}\alpha^{b_{i}}\delta\left(\mathbf{r}'-\mathbf{r}'^{b_{i}}\left(t'\right)\right)v_{x}'^{b_{i}}\left(t'\right)\\
j'{}_{x}\left(\mathbf{r}',t'\right)+V\varrho'\left(\mathbf{r}',t'\right) & = & \sum_{i=1}^{n}\alpha^{b_{i}}\delta\left(\mathbf{r}'-\mathbf{r}'^{b_{i}}\left(t'\right)\right)v_{x}'^{b_{i}}\left(t'\right)\nonumber \\
 &  & +V\sum_{i=1}^{n}\alpha^{b_{i}}\delta\left(\mathbf{r}'-\mathbf{r}'^{b_{i}}\left(t'\right)\right)\\
j'{}_{y}\left(\mathbf{r}',t'\right) & = & \sum_{i=1}^{n}\alpha^{b_{i}}\delta\left(\mathbf{r}'-\mathbf{r}'^{b_{i}}\left(t'\right)\right)v_{y}'^{b_{i}}\left(t'\right)\\
j'{}_{z}\left(\mathbf{r}',t'\right) & = & \sum_{i=1}^{n}\alpha^{b_{i}}\delta\left(\mathbf{r}'-\mathbf{r}'^{b_{i}}\left(t'\right)\right)v_{z}'^{b_{i}}\left(t'\right)
\end{eqnarray}
Solving these linear equations for $\varrho'$ and $\mathbf{j}'$
we obtain \eqref{eq:add1'}--\eqref{eq:add2'}.
\end{proof}
Combining all the results we obtained in Theorems~\ref{thm:-E2'}--\ref{thm:superpos'},
we have
\begin{eqnarray}
\nabla\cdot\mathbf{E}'\left(\mathbf{r}',t'\right) & = & \sum_{i=1}^{n}\alpha'^{b_{i}}\delta\left(\mathbf{r}'-\mathbf{r}'^{b_{i}}\left(t'\right)\right)\label{eq:MLE1'}\\
c^{2}\nabla\times\mathbf{B}'\left(\mathbf{r}',t'\right)-\partial_{t'}\mathbf{E}'\left(\mathbf{r}',t'\right) & = & \sum_{i=1}^{n}\alpha'^{b_{i}}\delta\left(\mathbf{r}'-\mathbf{r}'^{b_{i}}\left(t'\right)\right)\mathbf{v}'^{b_{i}}\left(t'\right)\label{eq:MLE2'}\\
\nabla\cdot\mathbf{B}'\left(\mathbf{r}',t'\right) & = & 0\label{eq:MLE3'}\\
\nabla\times\mathbf{E}'\left(\mathbf{r}',t'\right)+\partial_{t'}\mathbf{B}'\left(\mathbf{r}',t'\right) & = & 0\label{eq:MLE4'}\\
\gamma\left(\mathbf{v}'^{b_{i}}\left(t'\right)\right)\mathbf{a}'^{b_{i}}(t') & = & \pi'^{b_{i}}\Biggl\{\mathbf{E}'\left(\mathbf{r}'^{b_{i}}\left(t'\right),t'\right)\nonumber \\
 &  & +\mathbf{v}'^{b_{i}}\left(t'\right)\times\mathbf{B}'\left(\mathbf{r}'^{b_{i}}\left(t'\right),t'\right)\nonumber \\
 &  & -\mathbf{v}'^{b_{i}}\left(t'\right)\frac{\mathbf{v}'^{b_{i}}\left(t'\right)\mathbf{\cdot E}'\left(\mathbf{r}'^{b_{i}}\left(t'\right),t'\right)}{c^{2}}\Biggr\}\,\,\,\,\,\,\,\,\label{eq:MLE5'}\\
 &  & \,\,\,\,\,\,\,\,\,\,\,\,\,\,\,\,\,\,\,\,\,\,\,\,\,\,\,\,\,\,\,\,\,\,\,\,\,(i=1,2,\ldots n)\nonumber 
\end{eqnarray}

\section{Are the textbook transformation rules true?\label{sec:Are-the-textbook} }

Our main concern in this paper was: On what grounds can the textbook
transformation rules for the electrodynamic quantities---hence the
hypothesis of covariance itself, from which the rules are routinely
derived---be considered as empirically verified facts of the physical
world? Now everything is at hand to declare that the textbook transformation
rules are in fact true, at least in the sense that they are derivable
from the laws of ED in a single frame of reference---\emph{without
the prior assumption of covariance}. For, Theorems~\ref{thm:E'E}
and \ref{thm:B'B} show the well-known transformation rules for the
field variables. What Theorem~\ref{thm:aramtrafo} asserts is nothing
but the well-known transformation rule for charge density and current
density. Finally, Theorem~\ref{thm:-E2'} shows that a particle's
electric specific passive charge, active charge and electric rest
mass are invariant Lorentz scalars. And, of course, these results
make it possible to use the well-known covariant formulation of electrodynamics.

At this point, having ascertained the transformation rules, we can
recognize that equations (\ref{eq:MLE1'})--(\ref{eq:MLE5'}) are
nothing but equations (\ref{eq:MLE1})--(\ref{eq:MLE5}) expressed
in the primed variables. At the same time, (\ref{eq:MLE1'})--(\ref{eq:MLE5'})
are manifestly of the same form as (\ref{eq:MLE1})--(\ref{eq:MLE5}).
Therefore, we proved that the Maxwell--Lorentz equations are \emph{indeed}
covariant against the \emph{real} transformations of the kinematic
and electrodynamic quantities. In fact, we proved more: 
\begin{itemize}
\item The Lorentz equation of motion \eqref{eq:MLE5} is covariant separately.
\item The four Maxwell equations (\ref{eq:MLE1})--(\ref{eq:MLE4}) constitute
a covariant set of equations, separately from \eqref{eq:MLE5}.
\item (\ref{eq:MLE1})--(\ref{eq:MLE2}) constitute a covariant set of equations,
separately.
\item (\ref{eq:MLE3})--(\ref{eq:MLE4}) constitute a covariant set of equations,
separately.
\end{itemize}
None of these statements follows automatically from the fact that
(\ref{eq:MLE1})--(\ref{eq:MLE5}) form a covariant system of equations
(Gömöri and Szabó~2011).

It is of interest to notice that all these results hinge on the \emph{relativistic}
version of the Lorentz equation, in particular, on the ``relativistic
mass-formula''. Without factor $\gamma\left(\mathbf{v}{}^{b}\right)$
in (\ref{eq:MLE5}), the proper transformation rules were different
and the Maxwell equations were not covariant---against the proper
transformations.

\section*{Acknowledgment}

The related research was partly supported by the OTKA Foundation,
No.~K 68043.

\section*{Appendix 1}

It is assumed that space and time coordinates are defined in all inertial
frames of reference; that is, in an arbitrary inertial frame $K$,
space tags $\mathbf{r}\left(A\right)=\left(x\left(A\right),y\left(A\right),z\left(A\right)\right)\in\mathbb{R}^{3}$
and a time tag $t\left(A\right)\in\mathbb{R}$ are assigned to every
event $A$ \textemdash{}by means of some empirical operations. We
also assume that the assignment is mutually unambiguous, such that
there is a one to one correspondence between the space and time tags
in arbitrary two inertial frames of reference $K$ and $K'$; that
is, the tags $\left(x'\left(A\right),y'\left(A\right),z'\left(A\right)\right)$
can be expressed by the tags $\left(x\left(A\right),y\left(A\right),z\left(A\right)\right)$,
and vice versa. The concrete form of this functional relation is an
empirical question. In this paper, we will take it for granted that
this functional relation is the well-known Lorentz transformation

Below we recall the most important formulas we use. For the sake of
simplicity, we assume the usual situation: $K'$ is moving along the
$x$-axis with velocity $\mathbf{V}=\left(V,0,0\right)$ relative
to $K$, the corresponding axises are parallel and the two origins
coincide at time $0$. Throughout the paper we use the following notations:
$\gamma(\ldots)=\left(1-\frac{(\ldots)^{2}}{c^{2}}\right)^{-\frac{1}{2}}$
and $\gamma=\gamma(V)$.

The connection between the space and time tags of an event $A$ in
$K$ and $K'$ is the following:

\begin{eqnarray}
x'\left(A\right) & = & \gamma\left(x\left(A\right)-Vt\left(A\right)\right)\label{eq:LT1}\\
y'\left(A\right) & = & y\left(A\right)\\
z'\left(A\right) & = & z\left(A\right)\\
t'\left(A\right) & = & \gamma\left(t\left(A\right)-c^{-2}Vx\left(A\right)\right)\label{eq:LT4}
\end{eqnarray}
Let $A$ be an event on the worldline of a particle. For the velocity
of the particle at $A$ we have:

\noindent 
\begin{eqnarray}
v_{x}'\left(A\right) & = & \frac{v_{x}\left(A\right)-V}{1-c^{-2}v_{x}\left(A\right)V}\label{eq:LT5}\\
v_{y}'\left(A\right) & = & \frac{\gamma^{-1}v_{y}\left(A\right)}{1-c^{-2}v_{x}\left(A\right)V}\\
v_{z}'\left(A\right) & = & \frac{\gamma^{-1}v_{z}\left(A\right)}{1-c^{-2}v_{x}\left(A\right)V}\label{eq:LT8}
\end{eqnarray}
We also use the inverse transformation in the following special case:
\begin{eqnarray}
\mathbf{v}'\left(A\right)=\left(v',0,0\right) & \mapsto & \mathbf{v}\left(A\right)=\left(\frac{v'+V}{1+c^{-2}v'V},0,0\right)\label{eq:sebesseg3}\\
\mathbf{v}'\left(A\right)=\left(0,0,v'\right) & \mapsto & \mathbf{v}\left(A\right)=\left(V,0,\gamma v'\right)\label{eq:sebesseg1}
\end{eqnarray}

\noindent The transformation rule of acceleration is much more complex,
but we need it only for $\mathbf{v}'\left(A\right)=\left(0,0,0\right)$:

\noindent 
\begin{eqnarray}
a'_{x}\left(A\right) & = & \gamma^{3}a_{x}\left(A\right)\label{eq:gyorsulas1}\\
a'_{y}\left(A\right) & = & \gamma^{2}a_{y}\left(A\right)\label{eq:gyorsulas2}\\
a'_{z}\left(A\right) & = & \gamma^{2}a_{z}\left(A\right)\label{eq:gyorsulas3}
\end{eqnarray}
We will also need the $y$-component of acceleration in case of $\mathbf{v}'\left(A\right)=\left(0,0,v'\right)$:
\begin{equation}
a'_{y}\left(A\right)=\gamma^{2}a_{y}\left(A\right)\label{eq:gyorsulas4}
\end{equation}

\section*{Appendix 2}

There are two major versions of the textbook derivation of the transformation
rules for electrodynamic quantities from the hypothesis of covariance.
The first version follows Einstein's 1905 paper: 
\begin{lyxlist}{00.00.0000}
\item [{(1a)}] The transformation rules of electric and magnetic field
strengths are derived from the presumption of the covariance of the
homogeneous (with no sources) Maxwell equations.
\item [{(1b)}] The transformation rules of source densities are derived
from the transformations of the field variables. 
\item [{(1c)}] From the transformation rules of charge and current densities,
it is derived that electric charge is an invariant scalar. 
\end{lyxlist}
The second version is this: 
\begin{lyxlist}{00.00.0000}
\item [{(2a)}] The transformation rules of the charge and current densities
are derived from some additional \emph{assumptions}; typically from
one of the followings:

\begin{lyxlist}{00.00.0000}
\item [{(2a1)}] the invariance of electric charge (Jackson 1999, pp. 553--558)
\item [{(2a2)}] the current density is of form $\varrho\mathbf{u}(\mathbf{r},t)$,
where $\mathbf{u}(\mathbf{r},t)$ is a velocity field (Tolman 1949,
p. 85; M\o ller 1955, p. 140). 
\end{lyxlist}
\item [{(2b)}] The transformation of the field strengths are derived from
the transformation of $\varrho$ and $\mathbf{j}$ and from the presumption
of the covariance of the inhomogeneous Maxwell equations.
\end{lyxlist}
Unfortunately, with the only exception of (1b), none of the above
steps is completely correct. Without entering into the details, let
us mention that (2a1) and (2a2) both involve some further empirical
information about the world, which does not follow from the simple
assumption of covariance. Even in case of (1a) we must have the tacit
assumption that zero charge and current densities go to zero charge
and current densities during the transformation---otherwise the covariance
of the homogeneous Maxwell equations would not follow from the assumed
covariance of the Maxwell equations.

One encounters the next major difficulty in both (1a) and (2b): neither
the homogeneous nor the inhomogeneous Maxwell equations determine
the transformation rules of the field variables uniquely; $\mathbf{E}'$
and $\mathbf{B}'$ are only determined by $\mathbf{E}$ and $\mathbf{B}$
up to an arbitrary solution of the homogeneous equations (see also
Huang~2008). 

Finally, let us mention a conceptual confusion that seems to be routinely
overlooked in (1c), (2a1) and (2a2). There is no such thing as a simple
relation between the scalar invariance of charge and the transformation
of charge and current densities, as is usually claimed. For example,
it is meaningless to say that 
\begin{equation}
Q=\varrho\Delta W=Q'=\varrho'\Delta W'
\end{equation}
where $\Delta W$ denotes a volume element, and
\begin{equation}
\Delta W'=\gamma\Delta W
\end{equation}
Whose charge is $Q$, which remains invariant? Whose volume is $\Delta W$
and in what sense is that volume Lorentz contracted? In another form,
in (2a2), whose velocity is $\mathbf{u}(\mathbf{r},t)$?

\section*{References}
\begin{lyxlist}{00.00.0000}
\item [{Arthur~J.~W.~(2011):}] Understanding Geometric Algebra for Electromagnetic
Theory (IEEE Press Series on Electromagnetic Wave Theory), Wiley-IEEE
Press, Hoboken, NJ.
\item [{Bell,~J.S.~(1987):}] How to teach special relativity, in \emph{Speakable
and unspeakable in quantum mechanics}. Cambridge, Cambridge University
Press.
\item [{Einstein,~A~(1905):}] \foreignlanguage{ngerman}{Zur Elektrodynamik
bewegter Körper, \emph{Annalen der Physik}} \textbf{17}, 891. (On
the Electrodynamics of Moving Bodies, in H. A. Lorentz et al.,\emph{
The principle of relativity: a collection of original memoirs on the
special and general theory of relativity. }London, Methuen and Company
1923)
\item [{Frisch,~M.~(2005):}] \emph{Inconsistency, Asymmetry, and Non-Locality},
Oxford, Oxford University Press. 
\item [{Gr\o n,~\O .~and~V\o yenli,~K.~(1999):}] On the Foundation
of the Principle of Relativity, \emph{Foundations of Physics} \textbf{29},
pp. 1695-1733.
\item [{Gömöri,~M.~and~L.E.~Szabó~(2011):}] On the formal statement
of the special principle of relativity, preprint, (http://philsci-archive.pitt.edu/id/eprint/8783).
\item [{Hestenes~D.~(1966):}] \emph{Space-Time Algebra,} New York, Gordon
\& Breach.
\item [{Hestenes,~D\emph{.~}(2003):}] Spacetime physics with geometric
algebra, \emph{Am. J. Phys.} \textbf{71}, 691, DOI: 10.1119/1.1571836. 
\item [{Huang,~Young-Sea~(1993):}] Has the Lorentz-covariant electromagnetic
force law been directly tested experimentally?, \emph{Foundations
of Physics Letters} \textbf{6}, 257.\textbf{ }
\item [{Huang,~Young-Sea~(2008):}] Does the manifestly covariant equation
$\partial_{\alpha}\mathbf{A}^{\alpha}=0$ imply that $\mathbf{A}^{\alpha}$
is a four-vector?, \emph{Canadian J. Physics} \textbf{86}, pp. 699\textendash{}701
DOI: 10.1139/P08-012.
\item [{Huang,~Young-Sea~(2009):}] A new perspective on relativistic
transformation for Maxwell\textquoteright{}s equations of electrodynamics,
\emph{Physica Scripta} \textbf{79}, 055001 (5pp) DOI: 10.1088/0031-8949/79/05/055001.
\item [{Ivezi\'{c},~T.~(2001):}] ``True Transformations Relativity''
and Electrodynamics, \emph{Foundations of Physics} \textbf{31}, 1139.
\item [{Ivezi\'{c},~T.~(2003):}] The Proof that the Standard Transformations
of E and B Are not the Lorentz Transformations, \emph{Foundations
of Physics} \textbf{33}, 1339.
\item [{Jackson,~J.D.~(1999):}] \emph{Classical Electrodynamics (Third
edition).} Hoboken (NJ), John Wiley \& Sons.
\item [{Jammer,~M.~(2000):}] \emph{Concepts of Mass in Contemporary Physics
and Philosophy. }Princeton, Princeton University Press.
\item [{M\o ller~C.~(1955):}] \emph{The Theory of Relativity.} Oxford,
Clarendon Press. 
\item [{Muller,~F.~(2007):}] Inconsistency in Classical Electrodynamics?,
\emph{Philosophy of Science} \textbf{74}, pp. 253-277. 
\item [{Norton,~J.D.~(1993):}] General Covariance and the Foundations
of General Relativity: Eight Decades of Dispute, \emph{Reports on
Progress in Physics} \textbf{56}, 791.
\item [{Reichenbach,~H.~(1965):}] \emph{The Theory of Relativity and
A Priori Knowledge. }Berkeley and Los Angeles, University of California
Press.
\item [{Rohrlich,~F.~(2007):}] \emph{Classical Charged Particles. }Singapore,
World Scientific\emph{.}
\item [{Szabó,~L.E.~(2004):}] On the meaning of Lorentz covariance, \emph{Foundations
of Physics Letters} \textbf{17}, pp. 479--496.
\item [{Tolman,~R.C.~(1949):}] \emph{Relativity, Thermodynamics and Cosmology.}
Oxford, Clarendon Press.\end{lyxlist}

\end{document}